\documentclass[12pt]{article}

\usepackage{geometry}
\emergencystretch=\maxdimen
\hfuzz=\maxdimen
\usepackage{fourier}  
\usepackage{charter}  
\usepackage{bm}
\usepackage{dsfont}
\usepackage[mathscr]{eucal}

\usepackage{amsmath}
\usepackage{amsfonts}
\usepackage{amsthm}
\usepackage{amssymb}
\usepackage{setspace}
\usepackage{pdflscape}
\usepackage{actuarialsymbol}
\usepackage{tikz}
\usepackage{algorithm}
\usepackage{algorithmicx}

\usepackage{enumitem}
\setenumerate[1]{
	leftmargin=20pt,
	itemsep=-2pt,
	labelsep=0.5em,
	topsep=2pt}
\setitemize[1]{
	leftmargin=15pt,
	itemsep=1pt,
	labelsep=0.5em,
	topsep=2pt}

\usepackage{titling}
\pretitle{\begin{center}\begin{spacing}{1.01}\Large\bfseries}
		\posttitle{\par\end{spacing}
	\end{center}\vspace*{-1.25em}}

\usepackage{titletoc}
\usepackage[explicit]{titlesec}
\titleformat{\section}{\large\bfseries\boldmath\linespread{1.25}\selectfont}{\thesection.}{0.75em}{#1}
\titleformat{\subsection}{\bfseries\boldmath\linespread{0.25}\selectfont}{\thesubsection.}{0.75em}{#1}
\titleformat{\appendix}{\large\bfseries\linespread{1.25}\selectfont}{Appendix\theappendix.}{0.75em}{#1}

\newtheoremstyle{mystyle}{20pt}{20pt}{\itshape}{0cm}{\bfseries}{.}{0.75em}{}
\theoremstyle{mystyle}
\newtheorem{thm}{Theorem}

\newtheorem{corollary}{Corollary}

\newtheoremstyle{mystyle02}{20pt}{20pt}{}{0cm}{\bfseries}{.}{0.75em}{}
\theoremstyle{mystyle02}
\newtheorem{remark}{Remark}

\makeatletter \@addtoreset{equation}{section} \makeatother

\usepackage{apacite}
\usepackage{natbib}

\usepackage{lastpage}
\usepackage{fancyhdr}
\usepackage{graphicx}
\usepackage{epstopdf}
\usepackage{subfigure}
\usepackage{overpic}

\usepackage{caption}
\usepackage{color}
\usepackage{hyperref}
\hypersetup{colorlinks=true,linkcolor=blue,citecolor=blue}
\usepackage{bookmark}

\usepackage[hang]{footmisc}
\renewcommand\footnoterule{\kern20pt{\hrule width0.75in height0.5pt}\kern13pt}

\allowdisplaybreaks[4]

\usepackage{caption} 
\usepackage{subfigure}
\usepackage{booktabs}
\usepackage{longtable}
\usepackage{multirow}

\begin{document}
	
	\title{Computing Gerber-Shiu function in the classical risk model with 
	interest 
	using collocation method}
	\author{Zan Yu, ~\quad~Lianzeng Zhang\thanks{Corresponding author.}\\[1em]
		{\small School of Finance, Nankai University, Tianjin 300350,  China}}
	\date{\vspace*{-3em}}
	\maketitle

    \def\thefootnote{}
	\footnote{E-mail addresses: zhlz@nankai.edu.cn(L. 
		Zhang), yz3006@163.com(Z. Yu)}

	\begin{center}\begin{spacing}{1.05}
			\begin{minipage}[t]{15cm}
				\abstract{The Gerber-Shiu function is a classical 
				research topic in actuarial science. However, exact 
				solutions are only available in the literature for very 
				specific cases where the claim amounts follow distributions 
				such as the exponential distribution. This presents a 
				longstanding challenge, particularly from a computational 
				perspective. For the classical risk process in 
				continuous time, the Gerber-Shiu discounted penalty function 
				satisfies a class of Volterra integral equations. In this 
				paper, we use the collocation method to compute the 
				Gerber-Shiu function for risk model with interest. Our 
				methodology demonstrates that the function can be expressed as 
				a linear algebraic system, which is straightforward to 
				implement. One major advantage of our approach is that it does 
				not require any specific distributional assumptions on the 
				claim amounts, except for mild differentiability and 
				continuity conditions that can be easily verified. We also 
				examine the convergence orders of the collocation method. 
				Finally, we present several numerical examples to illustrate 
				the desirable performance of our proposed method.}\\[7pt]
			\end{minipage}
	\end{spacing}\end{center}
	
	\enlargethispage{-1.75em}
	\thispagestyle{empty}

\section{Introduction} 
Consider the classical risk (surplus) process  in continuous time 
$\{U(t)\}_{t \geq 0}$ 
as below
\begin{equation}
	U(t)=u+c t-\sum_{i=1}^{N(t)} X_{i}, \quad t \geq 0,
	\label{0.1}
\end{equation}
where $u \geq 0$ is the initial reserve, $N(t)$ is the number of claims up 
to time $t$ which follows a homogeneous Poisson process of parameter $\lambda 
>0$. The aggregate claim amount up to time $t$ is 
$$
Z(t)=\sum_{i=1}^{N(t)} X_{i},
$$ 
where the claim sizes  $\left\{X_{i},i=1,2,\ldots\right\}$ are  positive, 
independent and identically distributed random variables, with
the distribution function $F(x)$ and finite mean $\mu$. The 
premium rate $c$ satisfies
\begin{equation}
	c=\lambda \mu(1+\theta),
\end{equation}
where $\theta > 0$ is the premium loading factor. As usual, we assume that 
$X_i$  
and $\{N(t)\}_{t > 0}$ are 
independent.

The classical surplus process defined in \eqref{0.1} does not take into 
account any interest earnings on investment. Assume that the insurer receives 
interest on its surplus at a constant force 
$\delta$ per unit time. Then
the modified surplus process, namely $\{U_{\delta}(t)\}_{t \geq 0}$ , can be 
described by 
\begin{equation}
	U_\delta(t)=u \mathrm{e}^{\delta t}+c \sx*{\angl{t}}[(\delta)] -\int_0^t 
	\mathrm{e}^{\delta(t-x)} \mathrm{d} Z(x).
\end{equation}

The time of ruin is the first time that the risk process takes a negative 
value and is denoted by 
$$
\tau=\inf \{t \geq 0\mid U_{\delta}(t)<0\},
$$
where $\tau=\infty$ if $U_\delta(t) \geq 0$ for all $t \geq 0$. To study the 
time to ruin $\tau$, the surplus immediately 
before ruin $U(\tau-)$, and the deficit
at ruin $\left|U(\tau)\right|$ in the classial surplus process, 
\cite{Gerber1998} 
proposed an expected discounted
penalty function. In this 
paper, we study the modified
Gerber-Shiu discounted penalty 
function defined by
\begin{equation}
	\Phi_{\delta,\alpha}(u)=\mathbb{E}\left[e^{-\alpha \tau} 
	w\left(U_{\delta}(\tau-),\left|U_{\delta}(\tau)\right|\right) 
	\mathbf{1}(\tau<\infty) 
	\mid U_{\delta}(0)=u\right], \quad u \geq 0, \alpha \geq 0,
\end{equation}
where $\mathbf{1}(A)$ is the indicator function of event $A$, and $w:[0, 
\infty) 
\times[0, \infty) \longmapsto[0, \infty)$ is a measurable
penalty function. Here, we can interpret $e^{-\alpha \tau}$ as the 
`discounting factor'.

The Gerber-Shiu function has been a popular tool for actuarial researchers due 
to its broad applicability in representing a range of ruin-related quantities 
since it was introduced. In the past two decades, various stochastic processes 
have been employed to model the temporal evolution of surplus process using 
the Gerber-Shiu function. For instance, \cite{Gerber1998} studied the expected 
value of a discounted penalty function 
under the classical risk model, also known as the Cramér-Lundberg model. Over 
time, the Cram{\'e}r-Lundberg model has been extended in multiple directions 
in order to describe more accurately the stylized features of the surplus 
process in 
the real world. For example, the Sparre-Andersen model 
(\citealt{LANDRIAULT2008600}, \citealt{CHEUNG2010117}), the L{\'e}vy risk 
model 
(\citealt{doi:10.1080/10920277.2006.10597421}, 
\citealt{doi:10.1080/03461238.2011.627747}), and the Markov additive processes 
(\citealt{li_lu_2008}, \citealt{doi:10.1080/10920277.2010.10597599}) have all 
been 
studied for this purpose. 

A comprehensive review of Gerber-Shiu function and its variants, risk 
surplus models, and additional structural features, along with a wide range of 
analytical, semi-analytical, and asymptotic methods has been provided by 
\cite{HE20231}. While most of the existing literature has focused on exploring 
explicit solutions for Gerber-Shiu functions, this approach has limitations as 
it heavily depends on assumptions about the underlying claim size 
distribution. 
Only under a few types of
distributions, the Gerber-Shiu function has explicit expressions, such as 
exponentials, 
combinations of exponentials, Erlangs, and their mixtures.  Therefore, 
developing numerical methods for computing the 
Gerber-Shiu function is of great importance. 
\cite{mnatsakanov2008nonparametric} derived the Laplace transform of the 
survival probability, thereafter the survival probability can be subsequently 
obtained 
through some numerical inversion methods. This method was also applied by 
\cite{shimizu2011estimation, shimizu2012non} to compute the Gerber-Shiu 
discounted penalty function in the L{\'e}vy risk model and the perturbed 
compound Poisson risk model, respectively. The Gerber-Shiu function can also 
be evaluated approximately by truncating an infinite Fourier series 
(see e.g. \citealt{CHAU2015170}, \citealt{zhang_2017, zhang2017estimating}). 
In addition, 
\cite{zhang2018new, zhang2019estimating} develop
the Gerber-Shiu function on the Laguerre basis, and then compute the
unknown coefficients based on sample information on claim numbers and
individual claim sizes. \cite{wang2019computing} apply the frame duality 
projection method to compute the Gerber-Shiu function.

In this paper, we shall compute the Gerber-Shiu function in a risk 
model under interest force from a new perspective that is easy to implement. 
We remark that some results have been obtained by e.g. \cite{CAI2002389}, 
\cite{wu2007gerber} under interest force, but the results therein are mostly 
concerned with structural properties and general integral results, but no 
specific computing methods are given. Here we start directly from the solution 
of the integral equation to compute the Gerber-Shiu function. 

The rest of this 
paper is organized as 
follows. We first review some basic 
conclusions of the Gerber-Shiu function in Section \ref{Preliminaries on 
	Gerber–Shiu function}. 
Section \ref{section:Collocation methods for linear 
	second-kind VIEs} introduces the collocation method. In Section 
\ref{section:Convergence of collocation methods}, we discuss the convergence 
order 
of the collocation method. We present 
some numerical examples to illustrate the effectiveness of collocation method 
in Section \ref{section:Numerical results}. Some conclusions are shown in 
Section \ref{section:Concluding Comments}.

\section{Preliminaries on Gerber-Shiu function}
\label{Preliminaries on Gerber–Shiu function}
In this section, we present some necessary preliminaries on Gerber-Shiu 
function. Some of the results are borrowed from \cite{CAI2002389}. 
Throughout this paper, we condition on the time $t$, and on the amount 
of the first claim $x$. Thus,
\begin{equation}
	\begin{aligned}
		\Phi_{\delta,\alpha}(u)=&\int_{0}^{\infty}\lambda e^{-\lambda 
			t}\int_{0}^{\infty}\mathbb{E}\left[e^{-\alpha \tau} 
		w\left(U(\tau-),\left|U(\tau)\right|\right) \mathbf{1}(\tau<\infty) 
		\mid U_{\delta}(0)=u \right] \mathrm{d} F(x) \mathrm{d} t\\
		=&\int_0^{\infty} \lambda \mathrm{e}^{-(\lambda+\alpha) t} \int_0^{u 
			\mathrm{e}^{\delta t}+c \sx*{\angl{t}}[(\delta)]} \Phi_{\delta, 
			\alpha}\left(u \mathrm{e}^{\delta t}+c 
		\sx*{\angl{t}}[(\delta)]-x\right) \mathrm{d} F(x) \mathrm{d} t \\
		& +\int_0^{\infty} \lambda \mathrm{e}^{-(\lambda+\alpha) t} 
		\int_{u 
			\mathrm{e}^{\delta t}+c 
			\sx*{\angl{t}}[(\delta)]}^{\infty} 
		w\left(u \mathrm{e}^{\delta t}+c 
		\sx*{\angl{t}}[(\delta)], x-u \mathrm{e}^{\delta t}-c 
		\sx*{\angl{t}}[(\delta)]\right) \mathrm{d} 
		F(x) \mathrm{d} t\\
		=&\lambda(\delta u+c)^{(\lambda+\alpha) / \delta} 
		\int_u^{\infty}(\delta y+c)^{-((\lambda+\alpha) / 
			\delta)-1}\left(\int_0^y \Phi_{\delta, \alpha}(y-x) \mathrm{d} 
		F(x)+A(y)\right) \mathrm{d} y,
	\end{aligned}
	\label{1.1}
\end{equation}
where
\begin{equation*}
	A(t)=\int_{t}^{\infty}w(t,s-t) \mathrm{d} F(s).
\end{equation*}

Differentiating \eqref{1.1} with respect to $u$, we get
\begin{equation}
	\frac{\mathrm{d}}{\mathrm{d} u} \Phi_{\delta, 
		\alpha}(u)=\frac{\lambda+\alpha}{c+\delta u} \Phi_{\delta, 
		\alpha}(u)-\frac{\lambda}{c+\delta u}\left(\int_0^u \Phi_{\delta, 
		\alpha}(u-x) \mathrm{d} F(x)+A(u)\right).
	\label{1.2}
\end{equation}

Thus, integrating \eqref{1.2}, then performing integration by parts, we get,
$$
\Phi_{\delta, \alpha}(u)=\frac{c \Phi_{\delta, \alpha}(0)}{c+\delta 
	u}-\frac{\lambda}{c+\delta u} \int_0^u A(t) \mathrm{d} t+\int_0^u 
	K_{\delta, 
	\alpha}(u, t) \Phi_{\delta, \alpha}(t) \mathrm{d} t,
$$
where
$$
K_{\delta, \alpha}(u, t)=\frac{\delta+\alpha+\lambda (1-F(u-t))}{c+\delta u}.
$$

In particular, let $\alpha=0$, and denote that 
$\Phi_\delta(u)=\Phi_{\delta, 
	0}(u)$, $\Phi_\delta$ can be simplified to a Volterra integral equation 
of the second-kind:
\begin{equation}
	\Phi_\delta(u)=g(u)+\int_0^u K_\delta(u, t) \Phi_\delta(t) 
	\mathrm{d} t,
	\label{1.3}
\end{equation}
where 
\begin{equation}
	g(u)=\frac{c \Phi_\delta(0)}{c+\delta u}-\frac{\lambda}{c+\delta 
		u} \int_0^u A(t) \mathrm{d} t
\end{equation}


Therefore, for \eqref{1.3}, if $\Phi_{\delta}(0)$ is given, we can use 
numerical method to approximate $\Phi_{\delta}(u)$. Fortunately, 
\cite{CAI2002389} obtained $\Phi_{\delta}(0)$ by using Laplace 
transforms.  Let 
\begin{equation*}
	m_A=\int_{0}^{\infty}A(t)\mathrm{d} t
\end{equation*}
and
\begin{equation*}
	\phi_1(s)=\frac{1}{\mu}\int_{0}^{\infty}e^{-sx}\bar{F}(x)\mathrm{d} x.
\end{equation*}

Then \begin{equation}
	\begin{aligned}
		\Phi_\delta(0) =\frac{\lambda m_A}{\kappa_\delta} \int_{0}^{\infty} 
		\beta(\delta z) \exp \left(-c z+\lambda \mu \int_0^z \phi_1(\delta s) 
		\mathrm{d} s\right) \mathrm{d} z
	\end{aligned}
	\label{1.5}
\end{equation}
where 
$$
\kappa_\delta=c \int_{0}^{\infty} \exp \left(-c z+\lambda \mu \int_0^z 
\phi_1(\delta s) \mathrm{d} s\right) \mathrm{d} z
$$
and
$$
\beta(s)=\frac{1}{m_A}\int_{0}^{\infty}e^{-sx}A(x) \mathrm{d} x.
$$

As a special example, let $w(x_1,x_2)=1$, $\Phi_\delta(u)$ is expressed as the 
ruin probability for 
the surplus process \eqref{0.1}. $A(t)=1-F(t)$, $m_A=\mu$, 
$\beta(s)=\phi_1(s)$,
\eqref{1.5} can be simplified to 
$$
\Phi_{\delta}(0)=\frac{\kappa_{\delta}-1}{\kappa_{\delta}}.
$$
which is equivalent to Eq. (14) of \cite{SUNDT19957}.


\begin{remark}
	It is well known (see, for example, \citealt{brunner_2004}) that if $g$ 
	and 
	the kernel $K_{\delta}$ are both continuous, then the second-kind Volterra 
	integral 
	equation \eqref{1.3} has a unique solution.
	\label{remark.1}
\end{remark}

\section{The computation process of collocation method }
\label{section:Collocation methods for linear second-kind VIEs}

In this section, we first introduce the collocation method for calucating the 
general Volterra integral equation. Then we study how to use the collocation 
method to 
compute the Gerber-Shiu function.

Consider the general linear Volterra integral equation (VIE) of the 
second-kind as below
\begin{equation}
	y(t)=g(t)+\int_{0}^{t} K(t, s)y(s) \mathrm{d} s, \quad t \in I:=[0, T].
	\label{2.0}
\end{equation}

On the RHS of \eqref{2.0}, the linear Volterra integral operator $\mathcal{V}: 
C(I) \rightarrow C(I)$ is 
defined by
\begin{equation}
	(\mathcal{V} y)(t):=\int_{0}^{t} K(t, s)y(s) \mathrm{d} s, \quad t \in I,
	\label{2.1}
\end{equation}
where $K \in C(D)$ is some given function defined on $D$, here $D:=\{(t, s): 0 
\leq s \leq t \leq T\}$. 
Let $g \in C(I)$ be a given function. 
Use the notion of \eqref{2.1}, the Volterra integral equation can be written by
\begin{equation}
	y(t)=g(t)+(\mathcal{V} y)(t), \quad t \in I. \label{2.2}
\end{equation}

Assume the solution of \eqref{2.2} can be approximated by collocation in 
the (continuous) piecewise polynomial space
$$
S_{m-1}\left(I_{h}\right):=\left\{v:\left.v\right|_{\sigma_{n}} \in \pi_{m-1}, 
\quad 0 \leq n \leq N-1 \right\},
$$ 
where 
$$
I_{h}:=\left\{t_{i}=t_{i}^{(N)}: i =0,1,\ldots , N \right\}
$$
denotes a grid on the given interval  $I:=[0, T]$, with $t_0 =0$ and $t_N =T$. 
Here, $\pi_{m-1}$ is the 
set of (real) polynomials of degree of $m - 1$ (with $m> 1$), and  
set $\sigma_{0}:=\left[t_{0}, t_{1}\right]$ and 
$\sigma_{n}:=\left(t_{n}, t_{n+1}\right]$ for $n=1, \ldots , N-1$. 

The quantity $h_n:=t_{n+1}-t_{n}$ is called the diameter of the grid 
$I_{h}$. In particular, for the convenience of calculation, we use the uniform 
grid 
$I_{h}$, which means
$$
h_{n} \equiv h = \frac{T}{N}
$$ 
for all $0 \leq n \leq N-1$.

Define the set of collocation points,
\begin{equation}
	X_{h}:=\left\{t_{n, i}=t_{n}+c_{i} h:  n=0,1, \ldots, N-1, \; i = 
	1,2,\ldots, m 
	\right\}
\end{equation}
which is determined by the given grid $I_{h}$ and the given collocation 
parameters 
$\left\{c_i \right\}\subset [0,1]$ such that $0\leq c_1 \leq \cdots \leq c_m 
\leq 1$, the collocation solution $u_{h} \in S_{m-1}\left(I_{h}\right)$ is 
defined by the collocation equation corresponding to \eqref{2.2},
\begin{equation}
	u_{h}(t)=g(t)+\left(\mathcal{V} u_{h}\right)(t)
	, \quad t \in X_{h}. 
	\label{2.3}
\end{equation}

Using the local Lagrange basis functions with set $\left\{c_i \right\}$,
$$
L_{i}(\theta):=\prod_{k \neq i}^{m} \frac{\theta-c_{k}}{c_{i} - c_{k}}, \quad 
\theta \in [0,1]
$$
and
$$
U_{n, i}:=u_{h}(t_{n, i})=u_{h}\left(t_{n}+c_{i} h \right), \quad i=1, \ldots, 
m.
$$ 

For any $t=t_{n}+\theta h$ on the subinterval $\sigma_{n}:= (t_{n}, 
t_{n+1}] $, the local represention $u_{h}$  can be written as a interpolation 
function:
\begin{equation}
	u_{h}(t)=u_{h}\left(t_{n}+\theta h \right)=\sum_{i=1}^{m} L_{i}(\theta) 
	U_{n, i}, 
	\quad \theta \in(0,1].
	\label{2.5}
\end{equation}

Thus, the collocation equation \eqref{2.3} 
assumes the form 
\begin{equation*}
	\begin{aligned}
		U_{n, i}&=g(t_{n, i})+\int_{0}^{t_{n,i}} K(t_{n, i}, s) 
		u_{h}(s) \mathrm{d} s
		\\
		&=g(t_{n, i})+\int_{0}^{t_{n}} K(t_{n, i}, s) u_{h}(s) \mathrm{d} 
		s+\int_{t_{n}}^{t_{n,i}} K(t_{n, i}, s) u_{h}(s) \mathrm{d} s
		\\
		&=g(t_{n, i})+\int_{0}^{t_{n}} K(t_{n, i}, s) u_{h}(s) \mathrm{d} s+ h 
		\int_{0}^{c_{i}} K\left(t_{n, i}, t_{n}+s h \right) u_{h}\left(t_{n}+s 
		h \right) \mathrm{d} s
	\end{aligned}
\end{equation*}

Consider 
\begin{equation}
	F_{n}(t):=\int_{0}^{t_{n}} K(t, s) u_{h}(s) \mathrm{d} 
	s=\sum_{l=0}^{n-1} 
	h  \int_{0}^{1} K\left(t, t_{l}+s h \right) 
	u_{h}\left(t_{l}+s h \right) \mathrm{d} s
	\label{2.7}
\end{equation}

If we set $t=t_{n, i}$ in
\eqref{2.7} and employ the local 
representation \eqref{2.5}, we may write 
$$
\begin{aligned}
	F_{n}\left(t_{n, i}\right) &=\sum_{l=0}^{n-1} h \int_{0}^{1} 
	K\left(t_{n, i}, t_{l}+s h \right) u_{h}\left(t_{l}+s h \right) \mathrm{d} 
	s 
	\\
	&=\sum_{l=0}^{n-1} h \sum_{j=1}^{m}\left(\int_{0}^{1} K\left(t_{n, i}, 
	t_{l}+s h \right) L_{j}(s) d s\right) U_{l, j}
\end{aligned}
$$

So on the subinterval $ (t_{n}, t_{n+1}] $, the collocation 
equation \eqref{2.3} can be rewritten as 
\begin{equation}
	\begin{aligned}
		U_{n, i}=&g\left(t_{n, i}\right)+\sum_{l=0}^{n-1} h 
		\sum_{j=1}^{m}\left(\int_{0}^{1} K\left(t_{n, i}, 
		t_{l}+s h \right) L_{j}(s) \mathrm{d} s\right) U_{l, j}\\
		 &\quad+h 
		\sum_{j=1}^{m}\left(\int_{0}^{c_{i}} K\left(t_{n, i}, t_{n}+s h 
		\right) 
		L_{j}(s) \mathrm{d} s\right) U_{n, j}\\
	\end{aligned}
\end{equation}

Let $\mathbf{U}_{n}:=\left(U_{n, 1}, \ldots, U_{n, m}\right)^{T}$ , 
$\mathbf{g}_{n}:=\left(g\left(t_{n, 1}\right), \ldots, g\left(t_{n, 
	m}\right)\right)^{T}$, and define the matrices in 
$L\left(\mathbb{R}^{m}\right)$,
\begin{equation}
	B_{n}^{(l)}:=\left(\begin{array}{c}
		\int_{0}^{1} K\left(t_{n, i}, t_{l}+s h \right) L_{j}(s) \mathrm{d} s 
		\\
	\end{array}\right)_{m\times m}, \quad 
	0 \leq l \leq n-1, \quad i, j=1,	\ldots, m 
	\label{2.8}
\end{equation}
and
\begin{equation}
	B_{n}:=\left(\begin{array}{c}
		\int_{0}^{c_{i}} K\left(t_{n, i}, t_{n}+s h \right) L_{j}(s) 
		\mathrm{d} s \\
	\end{array}\right)_{m\times m}, \quad i, j=1, \ldots, m 
	\label{2.9}
\end{equation}

The linear algebraic system for  $\mathbf{U}_{n} \in \mathbb{R}^{m}$ can be 
written compactly as
\begin{equation}
	\left[\mathcal{I}_{m}-h B_{n}\right] 
	\mathbf{U}_{n}=\mathbf{g}_{n}+\mathbf{G}_{n}, 
	\quad n=0,1, \ldots, N-1
	\label{2.11}
\end{equation}
with
$$
\mathbf{G}_{n}:=\left(F_{n}\left(t_{n, 1}\right), \ldots, F_{n}\left(t_{n, 
	m}\right)\right)^{T}=\sum_{l=0}^{n-1} h B_{n}^{(l)} \mathbf{U}_{l}
$$
Here, $\mathcal{I}_{m}$ denotes  the identity matrix in 
$L\left(\mathbb{R}^{m}\right)$.

Since the kernel $K$ of the Volterra operator $ \mathcal{V}$ is continuous on 
$D$, the elements of the matrices $B_{n}$ are all bounded. By Neumann Lemma 
(\citealt{ortega1990numerical}), the inverse of matrix  
$\mathcal{I}_{m}-h B_{n}$ exists whenever $h \left\|B_{n}\right\|<1$ for 
matrix norm. This clearly holds whenever $h$ is sufficiently small. So for any 
grid $I_{h}$ with grid diameter $h$, each of the linear algebraic 
systems \eqref{2.11} has a unique solution $\mathbf{U}_{n}$. Hence the 
collocation equation \eqref{2.3} defines a unique collocation solution $u_{h} 
\in S_{m-1}\left(I_{h}\right)$ for \eqref{2.0}, with local represention on 
$\sigma_{n}$ given by \eqref{2.5}.

Actually, the collocation solution $u_{h}$ to the Volterra integral equation 
on an interval $I$ is an element of some finite-dimensional function space 
(the collocation space) that satisfies the Volterra integral equation on an 
appropriate finite 
subset of points in $I$ (the set of collocation points).

We suppose that the distribution function of claim size, $F(x)$, is 
continuous. Obviously, $K_\delta$ and $g$ satisfies the assumption of 
Remark \ref{remark.1}. So, substituting $K_\delta$  in the linear algebraic 
systems, we can write $B_{n}^{(l)}$ and $B_{n}$ as follow:
$$
B_{n}^{(l)}=\left(\begin{matrix}
	\int_{0}^{1} \frac{\delta+\bar{F}(t_n+c_1 h-t_l-s h)}{c+\delta(t_n+c_1 h)} 
	 L_{1}(s) \mathrm{d} s &
	\cdots &\int_{0}^{1} \frac{\delta+\bar{F}(t_n+c_1 h-t_l-s 
	h)}{c+\delta(t_n+c_1 h)} L_{m}(s) 
	\mathrm{d} 
	s\\
\vdots & &\vdots
	\\
	\int_{0}^{1} \frac{\delta+\bar{F}(t_n+c_m h-t_l-s h)}{c+\delta(t_n+c_m h)} 
	L_{1}(s) \mathrm{d} s &
	\cdots & \int_{0}^{1} \frac{\delta+\bar{F}(t_n+c_m h-t_l-s 
		h)}{c+\delta(t_n+c_m h)} L_{m}(s) 
	\mathrm{d} 
	s
	\\
\end{matrix}\right)_{m\times m}, 
$$
and
$$
B_{n}=\left(\begin{matrix}
	\int_{0}^{c_1} \frac{\delta+\bar{F}(t_n+c_1 h-t_l-s h)}{c+\delta(t_n+c_1 
	h)} 
	L_{1}(s) \mathrm{d} s &
	\cdots &\int_{0}^{c_1} \frac{\delta+\bar{F}(t_n+c_1 h-t_l-s 
		h)}{c+\delta(t_n+c_1 h)} L_{m}(s) 
	\mathrm{d} 
	s\\
	\vdots & &\vdots
	\\
	\int_{0}^{c_m} \frac{\delta+\bar{F}(t_n+c_m h-t_l-s h)}{c+\delta(t_n+c_m 
	h)} 
	L_{1}(s) \mathrm{d} s &
	\cdots & \int_{0}^{c_m} \frac{\delta+\bar{F}(t_n+c_m h-t_l-s 
		h)}{c+\delta(t_n+c_m h)} L_{m}(s) 
	\mathrm{d} 
	s
	\\
\end{matrix}\right)_{m\times m}. 
$$

We can solve the corresponding linear algebraic system to approximate the 
Gerber-Shiu function. It follows from the formula \eqref{2.8} and \eqref{2.9} 
that we have to 
compute integrals. However, for most of interesting distribution functions, 
the integrals can not be computed explicitly. This may cause some trouble in 
our 
calculations, but the numerical examples in section \ref{section:Numerical 
results} show that even if numerical integration is used, collocation method 
still has quite good accuracy.

\section{Convergence orders of collocation method} 
\label{section:Convergence of collocation methods}
In this section, starting from the error of the interpolation function, we 
focus on the convergence of Equation \eqref{2.3} when the kernel function is 
smooth. The results show that the convergence is closely related to the number 
of collocation parameters.

For any continuous function $y$, the error between $y$ and the Langrange 
interpolation polynomial with the given points $\left\{x_j \right\}$  is 
defined by
$$
e_{m}(y ; t):=y(t)-\sum_{j=1}^{m} L_{j}(t) y\left(x_{j}\right), \quad t 
\in[a, b]
$$ 
where $\left\{x_j \right\}\subset [a,b]$ such that $a\leq x_1 \leq \cdots \leq 
x_m \leq b$. \\

Here, we briefly review an important special case of the celebrated Peano 
Kernel Theorem.
\begin{thm}
	Assume that $y \in C^{m}[a, b]$. Then for given knots $a\leq x_1 < \cdots 
	< x_m \leq b$, $e_{m}(y ; t)$ 
	satisfies the integral representation
	\begin{equation}
		e_{m}(y ; t)=\int_{a}^{b} K_{m}(t, s) y^{(m)}(s) \mathrm{d} s, \quad t 
		\in[a, 
		b],
	\end{equation}
	where the Peano kernel $K_m$ is given by
	$$
	K_{m}(t, s):=\frac{1}{(m-1) !}\left\{(t-s)_{+}^{m-1}-\sum_{k=1}^{m} 
	L_{k}(t)\left(x_{k}-s\right)_{+}^{m-1}\right\}.
	$$ 
	and $y^{(m)}$ denotes the $m$th derivative.
	\label{Peano's Theorem}
\end{thm}

\begin{proof}
	Proofs of this important result may be found for example in 
	\cite{powell1981approximation} or \cite{stroud2012numerical}.
\end{proof}
On the subinterval $ (t_{n}, t_{n+1}] $, let $x_{j}=t_{n,j}=t_{n}+c_{j} h$, 
the 
interpolation error can be 
written as 
$$
e_{m}(y ; \theta):=y(t_{n}+\theta h)-\sum_{j=1}^{m} L_{j}(\theta) 
y\left(t_{n,j}\right), \quad \theta 
\in[0, 1].
$$ 
Thus, by theorem \ref{Peano's Theorem}, we can get the following corollary:
\begin{corollary} \label{Corollary}
	Under the assumptions of Theorem \ref{Peano's Theorem} and with 
	$[a,b]=[t_n,t_{n+1}]$, $t=t_n+\theta h \quad (\theta \in [0,1], 
	h:=t_{n+1}-t_{n}), \quad x_j=t_n+c_j h \quad (j=1, \ldots ,m)$ the 
	interpolation error 
	\begin{equation}
		e_{m}(y ; t):=y\left(t_{n}+\theta h\right)-\sum_{j=1}^{m} 
		L_{j}(\theta) y\left(t_{n}+c_{j} h\right), \quad \theta \in[0,1]
	\end{equation}
	can be expressed in the form
	\begin{equation}
		e_{m}\left(y ; t \right)=h^{m} \int_{0}^{1} 
		K_{m}(\theta, z) y^{(m)}\left(t_{n}+z h\right) \mathrm{d} z, \quad 
		\theta 
		\in[0,1]
	\end{equation}
	where
	$$
	K_{m}(\theta, z):=\frac{1}{(m-1) 
		!}\left\{(\theta-z)_{+}^{m-1}-\sum_{k=1}^{m} 
	L_{k}(\theta)\left(c_{k}-z\right)_{+}^{m-1}\right\}.
	$$
\end{corollary}
Define $R_{m, n}(\theta):=\int_{0}^{1} K_{m}(\theta, z) y^{(m)}\left(t_{n}+z 
h \right) d z$, we may resort to Peano's 
Theorem (see Corollary \ref{Corollary}) to write 
\begin{equation}
	y\left(t_{n}+\theta h \right)=\sum_{j=1}^{m} L_{j}(\theta) Y_{n, 
		j}+h^{m} 
	R_{m, n}(\theta), \quad \theta \in[0,1], 
\end{equation}
and $Y_{n,j} = y(t_{n,j})$. 

Consider that $u_{h}\left(t_{n}+\theta h 
\right)=\sum_{j=1}^{m} L_{j}(\theta) 
U_{n, j}$,
then for any $\theta$, the error between $y$ and the collocation solutions 
$u_{h}$ is represented by
\begin{equation}
	e_{h}\left(t_{n}+\theta h \right)=\sum_{j=1}^{m} L_{j}(\theta) 
	\mathcal{E}_{n, 
		j}+h^{m} R_{m, n}(\theta), \quad \theta \in(0,1]
	\label{3.3}
\end{equation}
with the collocation error, $\mathcal{E}_{n, j}:=Y_{n, j}-U_{n, j}$.\\

Rewrite the error $e_{h}$ with the Volterra integral equation \eqref{2.0}:
$$
\begin{aligned}
	e_{h}(t)&=y(t)-u_{h}(t)\\
	&=\int_{0}^{t} K(t, s)(y(s)-u_{h}(s)) \mathrm{d} s\\
	&=\int_{0}^{t} K(t, s)e_{h}(s) \mathrm{d} s.
\end{aligned}
$$ 
Then, for the error on collocation points $ t_{n,i}$, we have the following 
equation
\begin{equation}
	e_{h}\left(t_{n, i}\right)=\left(\mathcal{V} e_{h}\right)\left(t_{n, 
		i}\right), \quad 0 \leq n \leq N-1, \quad i=1, \ldots, m.
	\label{3.6}
\end{equation}
Furthermore, for the right side of the equation \eqref{3.6}
$$
\begin{aligned}
	\left(\mathcal{V} e_{h}\right)\left(t_{n, i}\right)
	=& \int_{0}^{t_{n}} 
	K\left(t_{n, i}, s\right) e_{h}(s) \mathrm{d} s+h \int_{0}^{c_{i}} 
	K\left(t_{n, 
		i}, t_{n}+s h\right) e_{h}\left(t_{n}+s h\right) \mathrm{d} s \\
	=& \sum_{l=0}^{n-1} h \int_{0}^{1} K\left(t_{n, i}, t_{l}+s 
	h \right)\left(\sum_{j=1}^{m} L_{j}(s) \mathcal{E}_{l, 
		j}+h^{m} R_{m, l}(s)\right) \mathrm{d} s \\
	&+h \int_{0}^{c_{i}} K\left(t_{n, i}, t_{n}+s 
	h \right)\left(\sum_{j=1}^{m} L_{j}(s) \mathcal{E}_{n, j}+h^{m} 
	R_{m, n}(s)\right) \mathrm{d} s
\end{aligned}
$$
Hence, we have 
$$
\begin{aligned}
	\mathcal{E}_{n, i} =&h \sum_{j=1}^{m}\left(\int_{0}^{c_{i}} 
	K\left(t_{n, i}, t_{n}+s h \right) L_{j}(s) \mathrm{d} s\right) 
	\mathcal{E}_{n, 
		j} \\
	&+ \sum_{l=0}^{n-1} h \sum_{j=1}^{m}\left(\int_{0}^{1} 
	K\left(t_{n, i}, t_{l}+s h \right) L_{j} (s) \mathrm{d} s\right) 
	\mathcal{E}_{l, j} \\
	&+\sum_{l=0}^{n-1} h^{m+1} \int_{0}^{1} K\left(t_{n, i}, 
	t_{l}+s h\right) R_{m, l}(s) \mathrm{d} s \\
	&+h^{m+1} \int_{0}^{c_{i}} K\left(t_{n, i}, t_{n}+s h \right) R_{m, 
		n}(s) \mathrm{d} s \quad i=1, \ldots, m
\end{aligned}
$$

Let $\mathcal{E}_{n}:=\left(\mathcal{E}_{n, 1}, \ldots, 
\mathcal{E}_{n, m}\right)^{T} \in \mathbb{R}^{m}$, we obtain a system of 
linear equations about the collocation error
\begin{equation}
	\left[\mathcal{I}_{m}-h B_{n}\right] 
	\mathcal{E}_{n}=\sum_{l=0}^{n-1} h B_{n}^{(l)} 
	\mathcal{E}_{l}+\sum_{l=0}^{n-1} h^{m+1} 
	\rho_{n}^{(l)}+h^{m+1} \rho_{n}, \quad 0 \leq n \leq N-1
	\label{3.7}
\end{equation}
The vectors $\rho_{n}^{(l)}$ and $\rho_{n}$ in $\mathbb{R}^{m}$ are
$$
\rho_{n}^{(l)}:=\left(\int_{0}^{1} K\left(t_{n, i}, t_{l}+s 
h \right) R_{m, l}(s) \mathrm{d} s, \right)^{T}, 
\quad l< n
$$
and
$$
\rho_{n}:=\left(\int_{0}^{c_{i}} K\left(t_{n, i}, t_{n}+s h \right) R_{m, 
	n}(s) \mathrm{d} s,  \right)^{T}
$$
with $ i=1, \ldots, m$.

We are now ready to give the convergence theorem.
\begin{thm}
	Assume the Volterra integral equation satisfies \\
	(a) $K \in C^{m}(D)$ and $g \in C^{m}(I)$.\\
	(b) $u_{h} \in S_{m-1} \left(I_{h}\right)$ is the collocation solution to 
	\eqref{2.2} defined by \eqref{2.3}.
	Then 
	\begin{equation}
		\left\|y-u_{h}\right\|_{\infty}:=\sup _{t \in 
			I}\left|y(t)-u_{h}(t)\right| 
		\leq C\left\|y^{(m)}\right\|_{\infty} h^{m}
	\end{equation}
	holds for any sets $X_{h}$ of collocation points with $0 \leq c_{1} \leq 
	\cdots \leq c_{m} \leq 1$. The constant $C$ depends on the $c_{i}$ but not 
	on $h$.
\end{thm}
\begin{proof}
	As the previous analysis, we conclude that the collocation error 
	$\mathcal{E}_{n}$ 
	is determined by \eqref{3.7}. And according 
	to the Neumann Lemma (\citealt{ortega1990numerical}), the  inverse of 
	matrix  $\mathcal{I}_{m}-h B_{n}$ exists, whenever 
	$h \left\|B_{n}\right\|<1$ for matrix norm.  In other words, for any grid 
	$I_{h}$, each matrix $\mathcal{I}_{m}-h B_{n}$ has a 
	uniformly bound inverse:
	$$
	\left\|\left(\mathcal{I}_{m}-h B_{n}\right)^{-1}\right\|_{1} \leq D_{0}, 
	\quad n=0,1, \ldots, N-1.
	$$
	
	Assume that $\left\|B_{n}^{(l)}\right\|_{1} \leq D_{1}$ for $0 \leq l<n 
	\leq N-1$, and
	$$
	\left\|\rho_{n}^{(l)}\right\|_{1} \leq m \bar{K} k_{m} M_{m}, \quad l<n, 
	\quad\left\|\rho_{n}\right\|_{1} \leq m \bar{K} k_{m} M_{m}.
	$$
	
	We define 
	$$
	M_{m}:=\left\|y^{(m)}\right\|_{\infty}, \quad k_{m}:=\max _{\theta 
		\in[0,1]} 
	\int_{0}^{1}\left|K_{m}(\theta, z)\right| \mathrm{d} z
	$$
	and
	$$
	\bar{K}:=\max _{t \in I} \int_{0}^{t}|K(t, s)| \mathrm{d} 
	s=\|\mathcal{V}\|_{\infty}
	$$
	Then, from (3.7)
	$$
	\left\|\mathcal{E}_{n}\right\|_{1} \leq D_{0} D_{1} \sum_{l=0}^{n-1} 
	h \left\|\mathcal{E}_{l}\right\|_{1}+D_{0}\left[m \bar{K} k_{m} M_{m} 
	\sum_{l=0}^{n-1} h^{m+1}+h^{m+1} m \bar{K} k_{m} M_{m}\right]
	$$
	and hence
	\begin{equation}
		\left\|\mathcal{E}_{n}\right\|_{1} \leq \gamma_{0} \sum_{l=0}^{n-1} 
		h \left\|\mathcal{E}_{l}\right\|_{1}+\gamma_{1} M_{m} h^{m}, 
		\quad n=0,1, \ldots, N-1,
		\label{3.9}
	\end{equation}
	where $\gamma_{0}:=D_{0} D_{1}$, $\gamma_{1}:=m D_{0} \bar{K} 
	k_{m}(T+h)$.\\ 
	Using Gronwall's inequality (\citealt{10.2307/2035962}), equation 
	\eqref{3.9} is 
	bounded by 
	$$
	\begin{aligned}
		\left\|\mathcal{E}_{n}\right\|_{1} & \leq \gamma_{1} M_{m} h^{m} \exp 
		\left(\gamma_{0} \sum_{l=0}^{n-1} h \right) \\
		& \leq \gamma_{1} M_{m} h^{m} \exp \left(\gamma_{0} T\right) 
		\quad(n=0,1, 
		\ldots, N-1)
	\end{aligned}
	$$
	In other words, there exists a constant $B<\infty$ such that, uniformly 
	for 
	$h\in (0,\bar{h})$,
	$$
	\left\|\mathcal{E}_{n}\right\|_{1} \leq B M_{m} h^{m}, \quad n=0,1, 
	\ldots, N-1.
	$$
	Setting $\Lambda_{m}:= \max_{j}\left\|L_{j}\right\|_{\infty}$, the local 
	error representation \eqref{3.3} is
	$$
	\left|e_{h}\left(t_{n}+\theta h \right)\right| \leq 
	\Lambda_{m}\left\|\mathcal{E}_{n}\right\|_{1}+h^{m} k_{m} M_{m} 
	\leq\left(\Lambda_{m} B+k_{m}\right) M_{m} h^{m}
	$$
	So for $\theta \in[0,1]$ and $n=0,1, \ldots, N-1$, we have
	$$
	\left\|e_{h}\right\|_{\infty} 
	\leq 
	C\left\|y^{(m)}\right\|_{\infty} h^{m},
	$$
	where the constant $C$ depends on the $c_{i}$ but not 
	on 
	$h$.	
\end{proof}
\begin{remark}
	For the collocation methods of Volterra integral equation, there are many 
	papers discussing the convergence under different conditions; see for 
	example
	\cite{diogo1994} and \cite{brunner_2017}. This paper only gives a rough 
	proof 
	of  the smooth kernel 
	function. For more detailed proof, refer to \citet{brunner_2004}.
\end{remark}

\section{Numerical results}\label{section:Numerical results}
In this section, we present some numerical examples to show that the 
collocation method is very efficient for computing the Gerber-Shiu function. 
All results are performed in MATLAB on Windows, with Intel(R) Core(TM) i7 
CPU, at 2.60 GHz and a RAM of 16 GB.
The collocation solution is now determined by the resulting 
system \eqref{2.11} and the local Langrange representation \eqref{2.5}. In all 
examples, we set $c=1.2$, $\lambda=1$ and $\delta=0.01$.
Here we consider the following claim size densities:
\begin{enumerate}
	\item[(1)] Exponential density: $f(x)=\mathrm{e}^{-x}, x>0$.
	\item[(2)] Erlang(2) density: $f(x)=4x\mathrm{e}^{-2x}, x>0 $.
	\item[(3)] Combination of exponentials density: $f(x)=3 \mathrm{e}^{-1.5 
		x}-3 \mathrm{e}^{-3 x}$.
\end{enumerate}
Note that these distributions have a common mean of $1$. We will compute the 
following three special Gerber-Shiu functions:  
\begin{enumerate}
	\item[(1)] Ruin probability with interest: 
	$\Phi_{\delta}(u)=\mathbb{P}\left(\tau<\infty \mid 
	U_{\delta}(0)=u\right)$, where  $w(x, y) 
	\equiv 1$.
	\item[(2)] Expected claim size causing ruin: 
	$\Phi_{\delta}(u)=\mathbb{E}\left[\left(U_{\tau-}+\left|U_\tau\right|\right)
	 1_{(\tau<\infty)} \mid U(0)=u\right]$, where  
	$w(x, y)=x+y$.
	\item[(3)] Expected deficit at ruin: 
	$\Phi_{\delta}(u)=\mathbb{E}\left[\left|U_\tau\right| 1_{(\tau<\infty)} 
	\mid 
	U(0)=u\right]$, where $ w(x, y)=y$.
\end{enumerate}
Throughout this section, we take $c_1=\frac{1}{3}$, $c_2=\frac{2}{3}$ when 
$m=2$ and $c_1=\frac{1}{3}$, $c_2=\frac{2}{3}$, $c_3=1$ when $m=3$, and we set 
the grid  $h=30/N$.   In Figure 
\ref{Fig.1}, we plot numerical result curves of exponential claim size density 
when $N$=4096. We conducted a comparison of our approach with other methods 
reported in the literature, and we find that the ruin probability we obtained 
were highly similar to those presented in the studies by \cite{zhang2018new}.
 In Figures \ref{Fig.2}--\ref{Fig.4}, we plot the relative errors under 
 different $N$ on 
the same picture to show variability bands and illustrate the stability of the 
procedures. We observe that the error are very close to each 
other when $N$ is large enough.  
\begin{figure}[htbp]
	\centering  
	\subfigure[]{
		\label{Fig.sub.1}
		\includegraphics[width=0.3\textwidth]{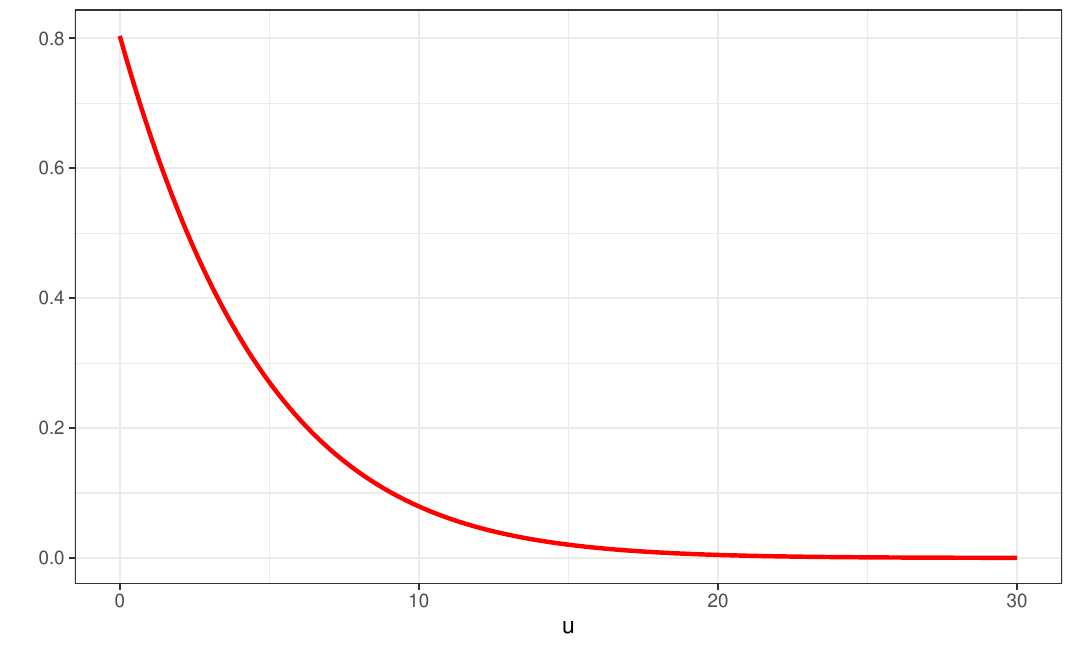}}
	\subfigure[]{
		\label{Fig.sub.2}
		\includegraphics[width=0.3\textwidth]{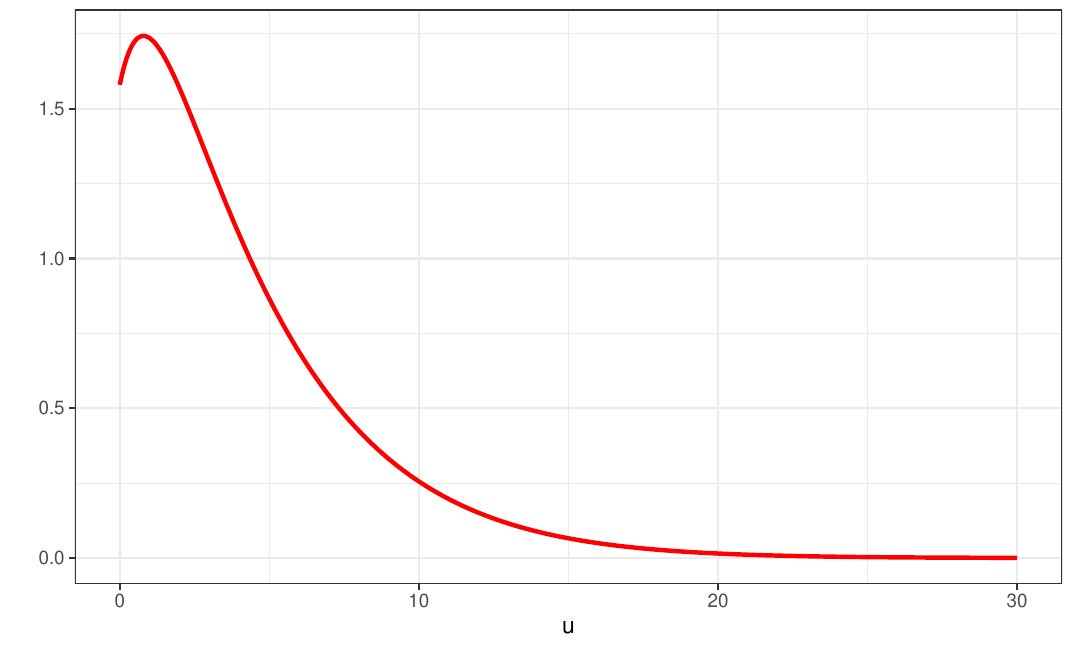}}
	\subfigure[]{
		\label{Fig.sub.3}
		\includegraphics[width=0.3\textwidth]{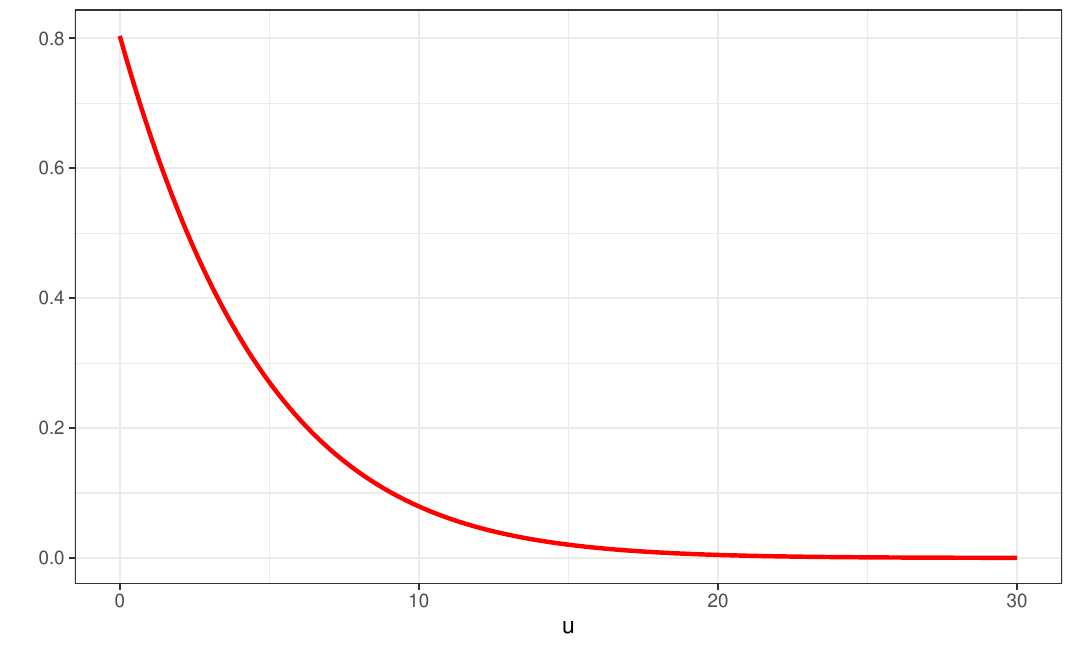}}
	\caption{Computing the Gerber-Shiu functions with Exponential claim size 
	density where $N=4096$. (a) ruin probability; (b) expected claim size 
	causing ruin; (b) expected deficit at ruin.}
	\label{Fig.1}
\end{figure}

\begin{figure}[htbp]
	\centering  
	\subfigure[]{
		\label{Fig2.sub.1}
		\includegraphics[width=0.3\textwidth]{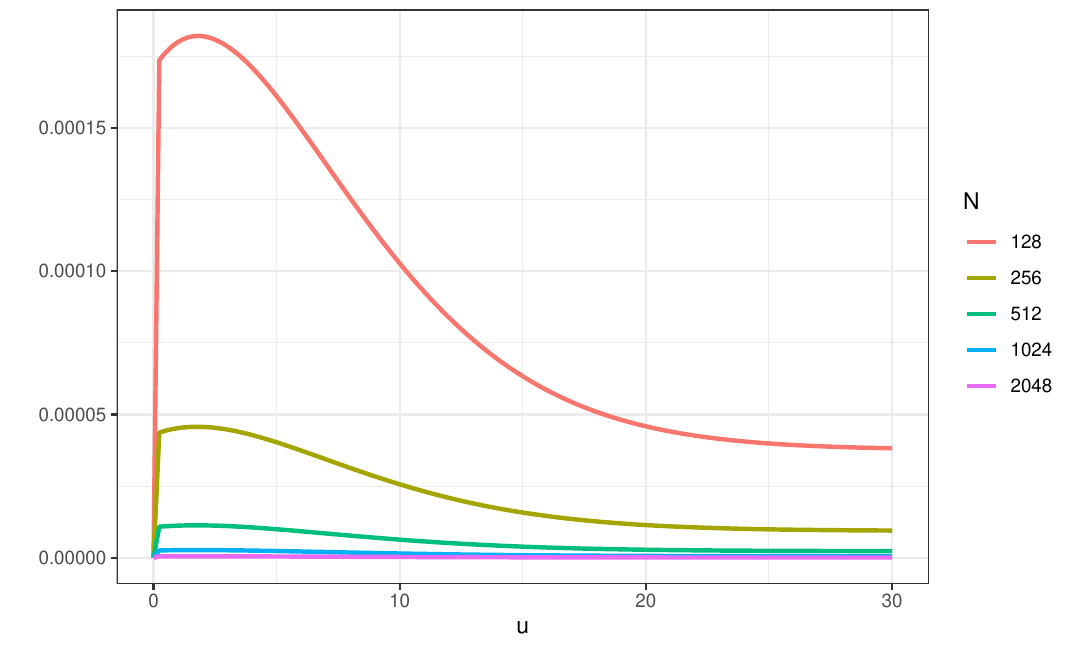}}
	\subfigure[]{
		\label{Fig2.sub.2}
		\includegraphics[width=0.3\textwidth]{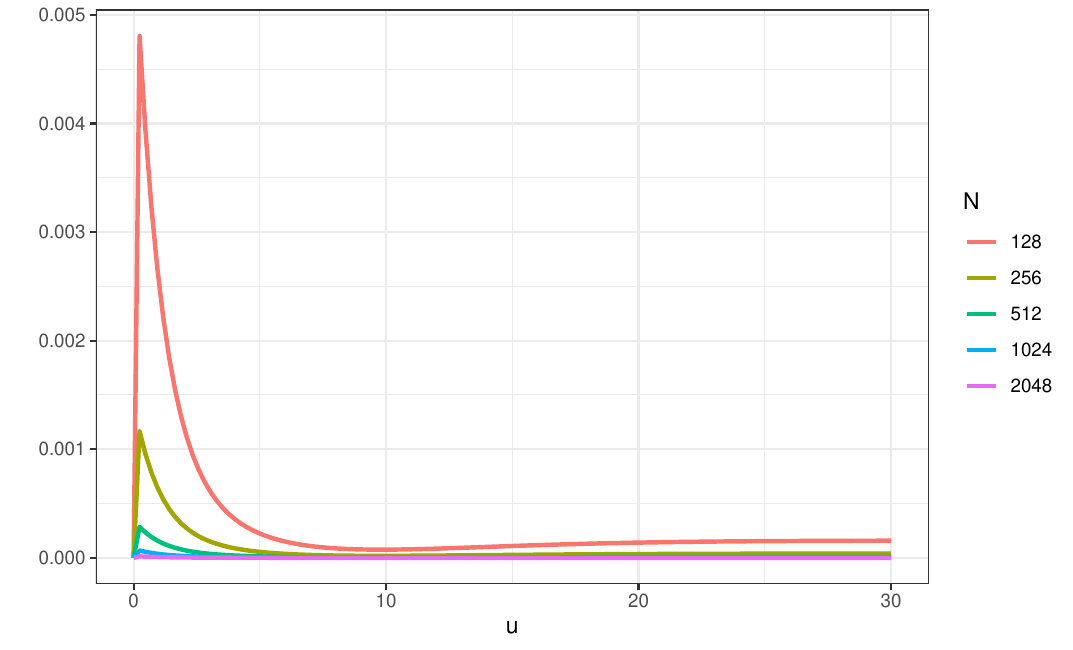}}
	\subfigure[]{
		\label{Fig2.sub.3}
		\includegraphics[width=0.3\textwidth]{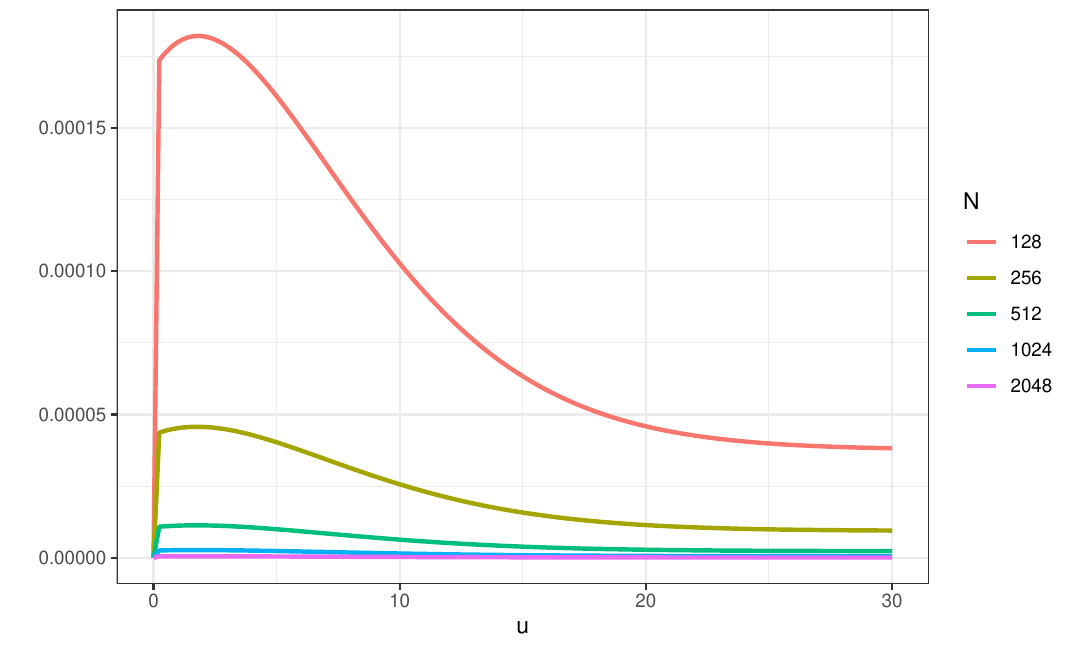}}
	\caption{Relative error of the Gerber-Shiu functions with Exponential 
	claim size 
		density. (a) ruin probability; (b) expected claim size 
		causing ruin; (b) expected deficit at ruin.}
	\label{Fig.2}
\end{figure}

\begin{figure}[htbp]
	\centering  
	\subfigure[]{
		\label{Fig3.sub.1}
		\includegraphics[width=0.3\textwidth]{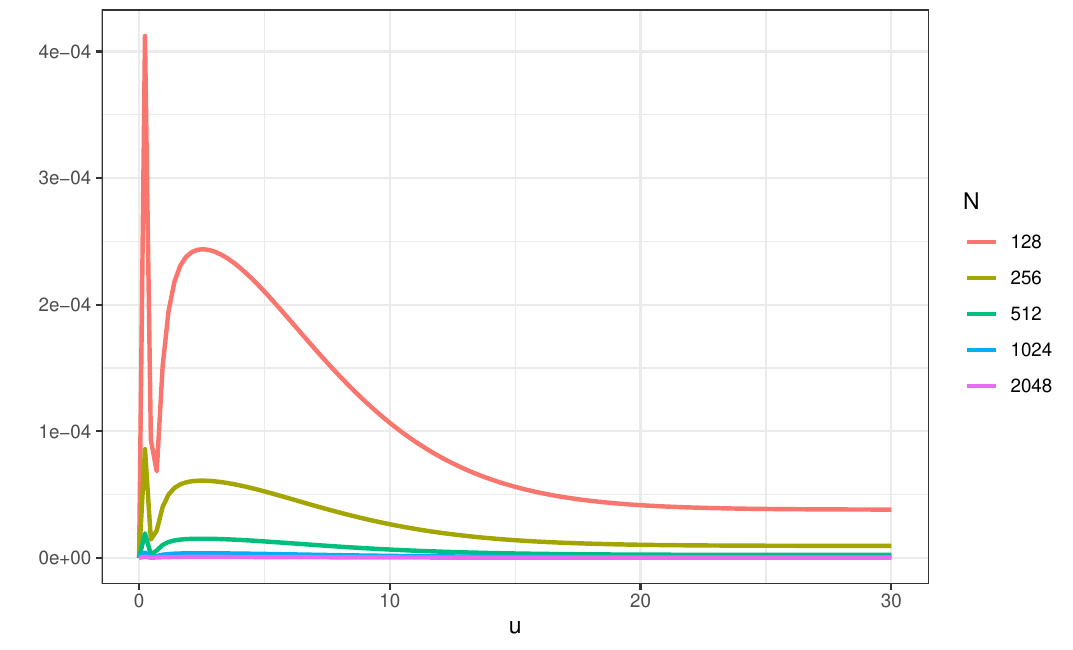}}
	\subfigure[]{
		\label{Fig3.sub.2}
		\includegraphics[width=0.3\textwidth]{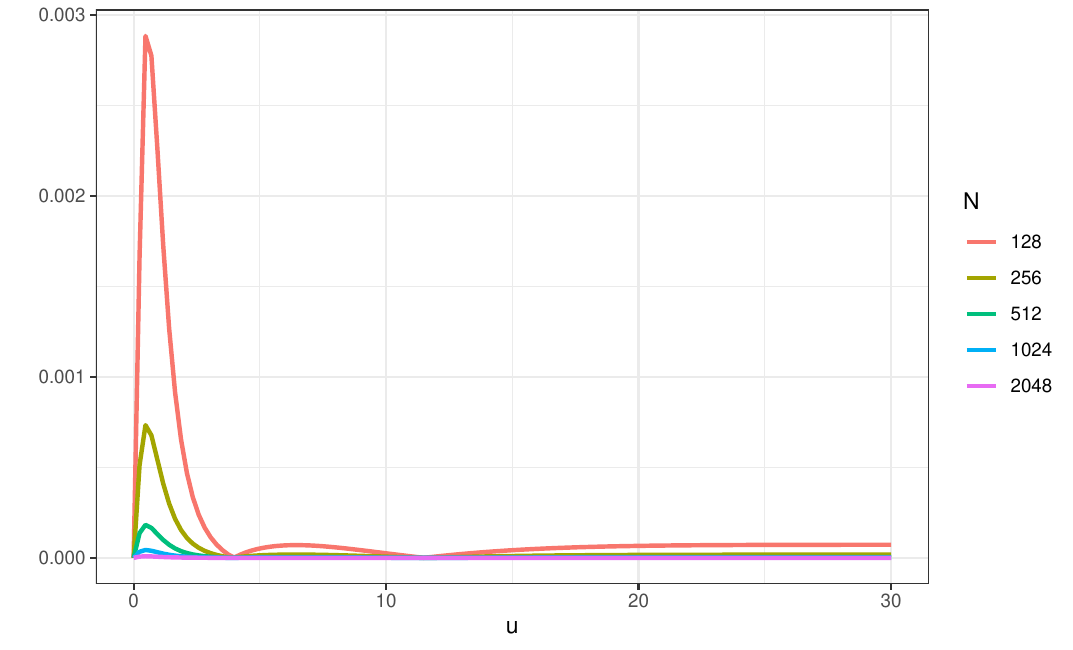}}
	\subfigure[]{
		\label{Fig3.sub.3}
		\includegraphics[width=0.3\textwidth]{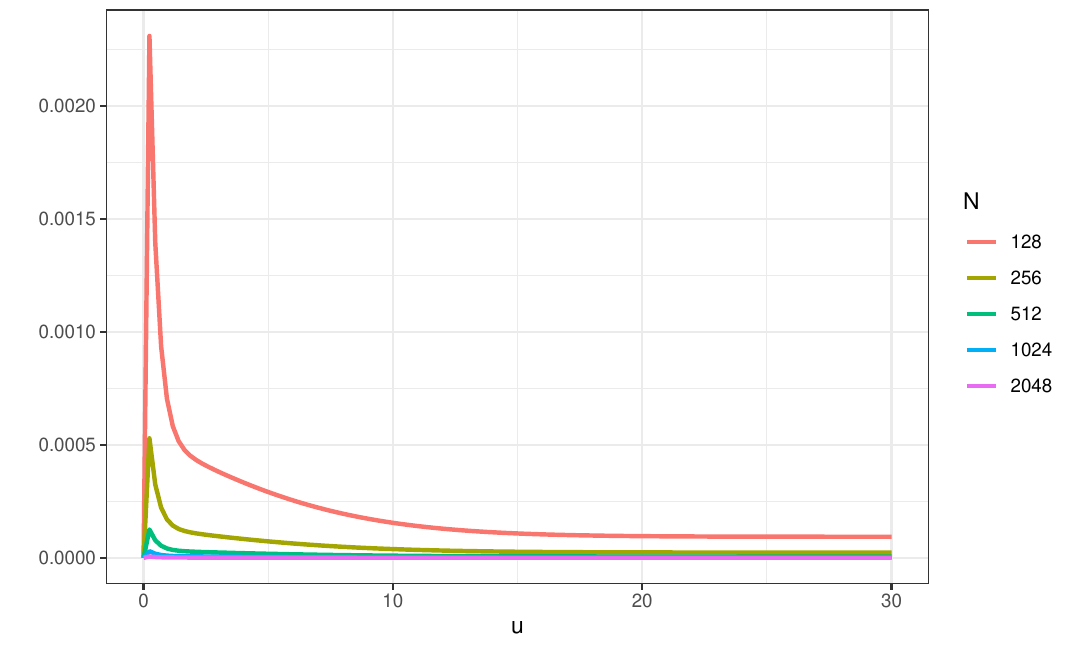}}
	\caption{Relative error of the Gerber-Shiu functions with Erlang 
		claim size 
		density. (a) ruin probability; (b) expected claim size 
		causing ruin; (b) expected deficit at ruin.}
	\label{Fig.3}
\end{figure}

\begin{figure}[htbp]
	\centering  
	\subfigure[]{
		\label{Fig4.sub.1}
		\includegraphics[width=0.3\textwidth]{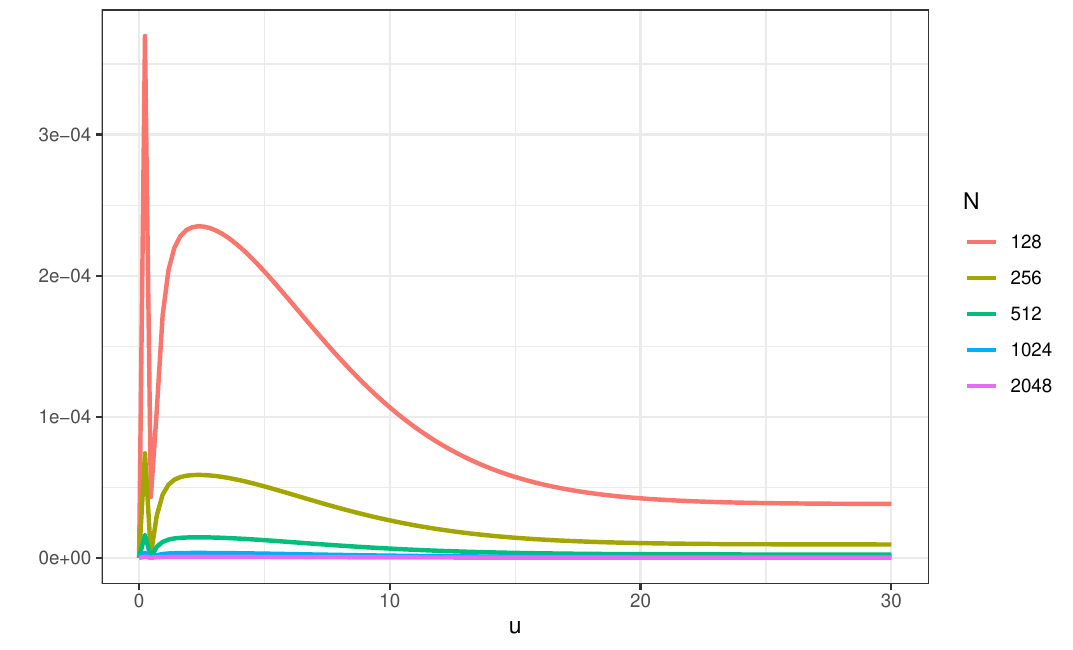}}
	\subfigure[]{
		\label{Fig4.sub.2}
		\includegraphics[width=0.3\textwidth]{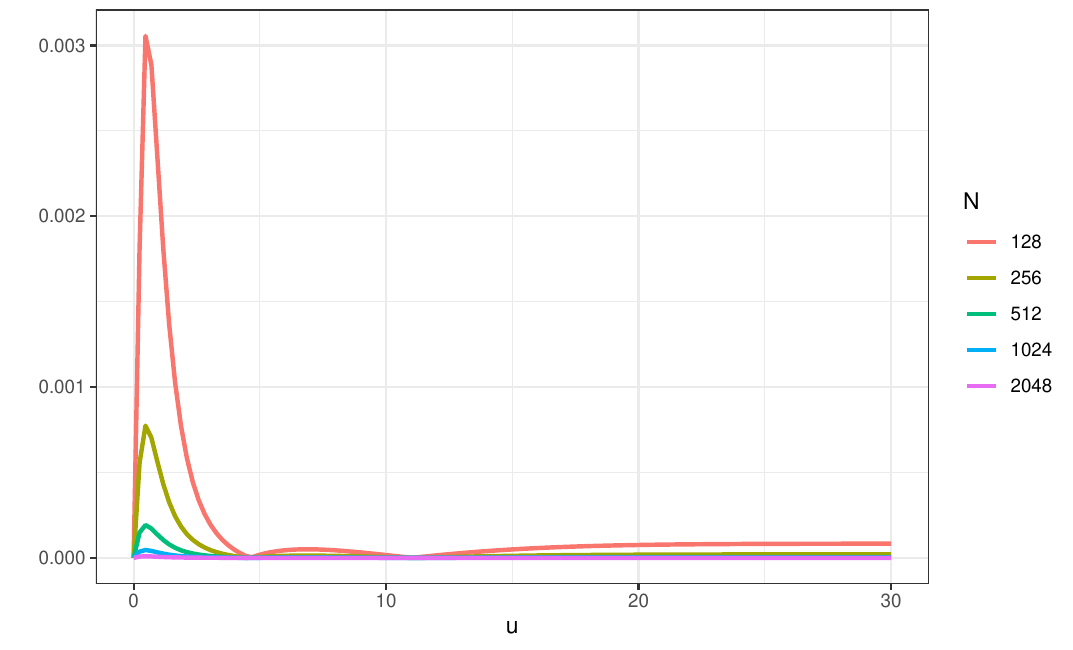}}
	\subfigure[]{
		\label{Fig4.sub.3}
		\includegraphics[width=0.3\textwidth]{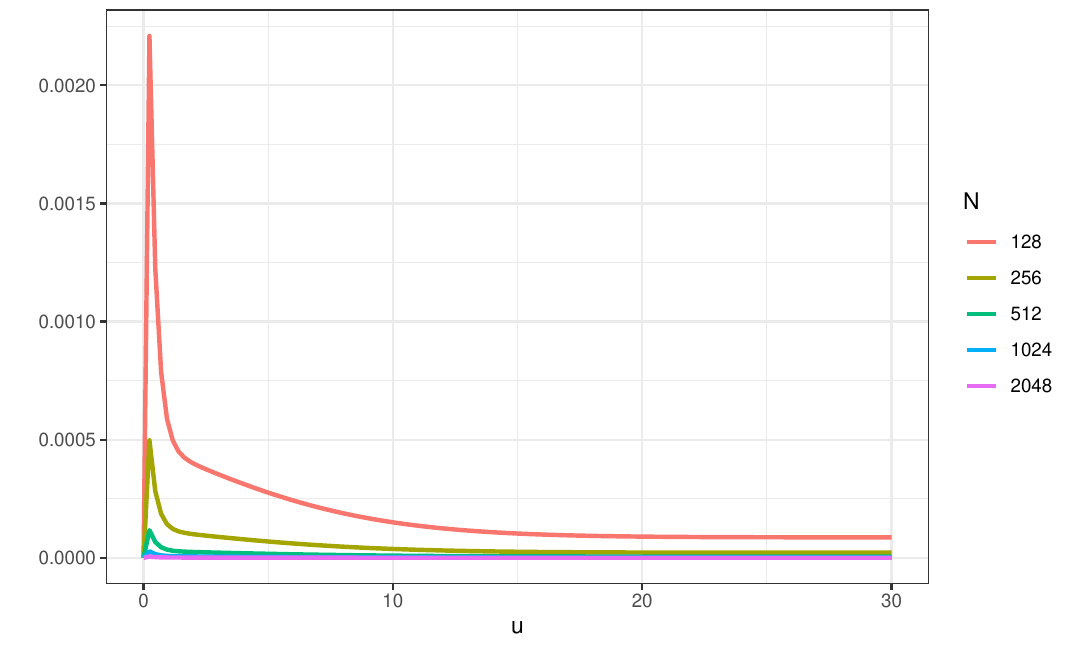}}
	\caption{Relative error of the Gerber-Shiu functions with combination of 
	exponentials density . (a) ruin probability; (b) expected claim size 
		causing ruin; (b) expected deficit at ruin.}
	\label{Fig.4}
\end{figure}

Next, in order to verify the result in section \ref{section:Convergence of 
collocation methods}, we compute the errors and convergences orders. For 
equations with analytical solutions, the error $E_{N}^1$ is given 
by
$
E_{N}^{1}=\sup _{t \in I}\left| u_{h}(t)-y(t) \right|.
$
While the analytical solution cannot be written out explicitly, its error is 
defined as
$
E_{N}^{2}=\left|u_{h}^{N}(T)-u_{h}^{2 N}(T)\right|.
$
At this time, the corresponding convergence order can be expressed as
$$
p_i =\log_{2}\left(\frac{E_{N}^{i}}{E_{2 N}^{i}}\right), \quad i=1,2.
$$

The errors and  convergence orders with $u=5$ are given in Table 
\ref{tab:Errors and convergences order for exponential distribution when 
	$m=2$}
 and \ref{tab:Errors and convergences order for exponential distribution when 
 	$m=3$}. We calculate $E_{N}^1$ and $p_1$ for Ruin 
Probability with interest. The exact ruin probability in the compound Poisson 
model with exponential claims, is calculated using the results in 
\cite{SUNDT19957}. We found that when $N$ is large enough (grid is small 
enough), it is difficult for the numerical solution to change for a given 
number of decimal places. For Expected claim 
size causing ruin and Expected deficit at ruin, $E_{N}^2$ and $p_2$ is 
computed. The results show that the global convergence order of the 
2-point collocation method is 2, and the global convergence order of the 
3-point collocation method is 3, which is consistent with the conclusion 
of Section \ref{section:Convergence of collocation methods}.
\begin{table}[htbp]
	\centering
	\caption{Errors and convergences order for exponential distribution when 
		$m=2$}
	\resizebox{\textwidth}{!}{
	    \begin{tabular}{lcccllcccllccc}
		\toprule
		\multicolumn{4}{c}{Ruin Probability(u=5)} &       & 
		\multicolumn{4}{c}{expected claim size causing ruin(u=5)} &       & 
		\multicolumn{4}{c}{expected deficit at ruin(u=5)} \\
		\midrule
		& Value & $E_{N}^{1}$ &   $p_1$    &       &       & Value & 
		$E_{N}^{2}$ &  $p_2$     
		&       &       & Value & $E_{N}^{2}$ & $p_2$ \\
		$N$=64  & 0.2705232  & 1.8019E-05 & \multicolumn{1}{c}{} &       & 
		N=64  
		& 0.8649623  &       &       &       & $N$=64  & 0.2705232  &       &  
		\\
		$N$=128 & 0.2705367  & 4.5111E-06 & \multicolumn{1}{c}{1.9979 } 
		&       
		& $N$=128 & 0.8649439  & 1.8335E-05 &       &       & $N$=128 & 
		0.2705367  
		& 1.3508E-05 &  \\
		$N$=256 & 0.2705401  & 1.1286E-06 & \multicolumn{1}{c}{1.9990 } 
		&       
		& $N$=256 & 0.8649394  & 4.5585E-06 & 2.0080  &       & $N$=256 & 
		0.2705401  & 3.3825E-06 & 1.9976  \\
		$N$=512 & 0.2705409  & 2.8225E-07 & \multicolumn{1}{c}{1.9995 } 
		&       
		& $N$=512 & 0.8649382  & 1.1364E-06 & 2.0041  &       & $N$=512 & 
		0.2705409  & 8.4635E-07 & 1.9988  \\
		$N$=1024 & 0.2705411  & 7.0575E-08 & \multicolumn{1}{c}{1.9997 } 
		&       
		& $N$=1024 & 0.8649380  & 2.8370E-07 & 2.0021  &       & $N$=1024 & 
		0.2705411  & 2.1168E-07 & 1.9994  \\
		$N$=2048 & 0.2705412  & 1.7645E-08 & \multicolumn{1}{c}{1.9999 } 
		&       
		& $N$=2048 & 0.8649379  & 7.0873E-08 & 2.0010  &       & $N$=2048 & 
		0.2705412  & 5.2930E-08 & 1.9997  \\
		\bottomrule
	\end{tabular}}
	\label{tab:Errors and convergences order for exponential distribution when 
		$m=2$}
\end{table}

\begin{table}[htbp]
	\centering
	\caption{Errors and convergences order for exponential distribution when 
		$m=3$}
	\resizebox{\textwidth}{!}{
	\begin{tabular}{lcccllcccllccc}
		\toprule
		\multicolumn{4}{c}{Ruin Probability(u=5)} &       & 
		\multicolumn{4}{c}{expected claim size causing ruin(u=5)} &       & 
		\multicolumn{4}{c}{expected deficit at ruin(u=5)} \\
		\midrule
		& Value & $E_{N}^{1}$ &   $p_1$    &       &       & Value & 
		$E_{N}^{2}$ &  $p_2$     
		&       &       & Value & $E_{N}^{2}$ & $p_2$ \\
		$N$=64  & 0.2705412  & 1.8915E-08 & \multicolumn{1}{c}{} &       & 
		$N$=64  
		& 0.8649388  &       &       &       & $N$=64  & 0.2705412  &       &  
		\\
		$N$=128 & 0.2705412  & 2.3592E-09 & \multicolumn{1}{c}{3.0032 } 
		&       
		& $N$=128 & 0.8649380  & 8.3565E-07 &       &       & $N$=128 & 
		0.2705412  
		& 1.6556E-08 &  \\
		$N$=256 & 0.2705412  & 2.9457E-10 & \multicolumn{1}{c}{3.0016 } 
		&       
		& $N$=256 & 0.8649379  & 1.0430E-07 & 3.0021  &       & $N$=256 & 
		0.2705412  & 2.0646E-09 & 3.0034  \\
		$N$=512 & 0.2705412  & 3.6801E-11 & \multicolumn{1}{c}{3.0008 } 
		&       
		& $N$=512 & 0.8649379  & 1.3028E-08 & 3.0011  &       & $N$=512 & 
		0.2705412  & 2.5777E-10 & 3.0017  \\
		$N$=1024 & 0.2705412  & 4.5994E-12 & \multicolumn{1}{c}{3.0002 } 
		&       
		& $N$=1024 & 0.8649379  & 1.6279E-09 & 3.0006  &       & $N$=1024 & 
		0.2705412  & 3.2202E-11 & 3.0009  \\
		$N$=2048 & 0.2705412  & 5.7476E-13 & \multicolumn{1}{c}{3.0004 } 
		&       
		& $N$=2048 & 0.8649379  & 2.0344E-10 & 3.0003  &       & $N$=2048 & 
		0.2705412  & 4.0250E-12 & 3.0001  \\
		\bottomrule
	\end{tabular}}
	\label{tab:Errors and convergences order for exponential distribution when 
		$m=3$}
\end{table}

\section{Concluding remarks}\label{section:Concluding Comments}
In this paper, we apply the collocation method to compute the Gerber-Shiu 
function in the classical risk model. Through several numerical 
illustrations, 
we find that this method is very efficient for the computations.

Not only for the Volterra integral equation, the collocation method can also 
be used in some integro-differential equations. So we can use this method to 
solve other ruin related problems in risk theory. For example, in 
\cite{SHELDONLIN2003551} and \cite{YUEN2007104}, the Gerber-Shiu function can 
be expressed as a specific class of integro-differential equation in the 
classical surplus process with a constant dividend barrier. 

For the mathematical aspects, we focused on the convergence results of 
collocation methods.  When the kernel $K$ is sufficiently regular, the 
convergence is closely related to the number of collocation parameters. 
However, if the 
distribution function $F(x)$ does not satisfy the required assumptions, the 
desired 
result could not be obtained. Finding a suitable method to improve the 
calculation accuracy is an open problem, when the kernel function is not 
smooth enough. We leave that aside for future 
investigation.

	
	
	
	

\section*{Acknowledgement}
The author would like to thank the anonymous reviewers for their valuable 
suggestions.

\bibliographystyle{apacite}

\bibliography{ref}
\addcontentsline{toc}{section}{References}

\clearpage
	
\end{document}